\def\vers{0}

\documentclass[runningheads]{llncs}

\usepackage{graphicx}
\usepackage{amsfonts}
\usepackage{amssymb}
\usepackage{algorithm}
\usepackage{algorithmic}
\usepackage{float}
\usepackage{enumerate}
\usepackage{pgfplots, tikz}
\usepackage{amsmath}

\newlength\myindent
\newenvironment{proofsketch}{%
\proof}{\endproof}

\spnewtheorem{thm}{Theorem}{\bf }{\it }
\spnewtheorem{prop}[thm]{Proposition}{\bf }{\it }
\spnewtheorem{prob}[thm]{Open Problem}{\bf }{\it }
\spnewtheorem{cor}[thm]{Corollary}{\bf }{\it }
\spnewtheorem{lem}[thm]{Lemma}{\bf }{\it }
\spnewtheorem{defn}[thm]{Definition}{\bf }{\rm }
\spnewtheorem{rem}[thm]{Remark}{\bf }{\rm }
\spnewtheorem{exmp}[thm]{Example}{\bf }{\rm }
\spnewtheorem{clm}[thm]{Claim}{\bf }{\it }
\spnewtheorem{qust}[thm]{Question}{\bf }{\it }
\spnewtheorem{nota}[thm]{Notation}{\bf }{\rm }

\def\proof{\noindent {\bf Proof.} }

\begin{document}

\title{A Faster Exact Algorithm to Count X3SAT Solutions\thanks{
Sanjay Jain and Frank Stephan are supported in part by the Singapore Ministry
of Education Tier 2 grant AcRF MOE2019-T2-2-121 / R146-000-304-112.
Further,
Sanjay Jain is supported in part by NUS grant number C252-000-087-001.
We thank the anonymous referees of CP2020 for several helpful comments.}}

\author{Gordon Hoi\,\inst{1}, Sanjay Jain\,\inst{1} and
  Frank Stephan\,\inst{1}\,\inst{2}}

\titlerunning{A Faster Exact Algorithm to Count X3SAT Solutions}
\authorrunning{G.~Hoi, S.~Jain and F.~Stephan}

\institute{School of Computing, National University of Singapore,
13 Computing Drive, Block COM1, Singapore 117417, Republic of Singapore \\
\email{e0013185@u.nus.edu, sanjay@comp.nus.edu.sg}
\and
Department of Mathematics, National University of Singapore,
10 Lower Kent Ridge Road, Block S17, 
Singapore 119076, Republic of Singapore \\
\email{fstephan@comp.nus.edu.sg}}

\maketitle        

\begin{abstract}
\noindent
The Exact Satisfiability problem, XSAT, is defined as the problem of finding a satisfying assignment to
a formula in CNF such that there is exactly one literal in each clause assigned to be  ``1" and the other literals in 
the same clause are set to ``0". If we restrict the length of each clause to be at most 3 literals, 
then it is known as the X3SAT problem. In this paper, we consider the problem of counting the 
number of satisfying assignments to the X3SAT problem, which is also known as \#X3SAT. 

The current state of the art exact algorithm to solve \#X3SAT is given by Dahll\"of, Jonsson and Beigel and runs in
$O(1.1487^n)$, where $n$ is the number of variables in the formula. 
In this paper, we propose an exact algorithm for the \#X3SAT problem that runs in
$O(1.1120^n)$ with very few branching cases to consider, by using a result from Monien and Preis to give us a 
bisection width for graphs with at most degree 3.

\end{abstract}

\section{Introduction}
\noindent 
Given a propositional formula $\varphi$ in conjunctive normal form (CNF), a common question to ask would be if
there is a satisfying assignment to $\varphi$. This is known as the satisfiability problem, or SAT. Many
other variants of the satisfiability problem have also been explored. An important variant is the Exact Satisfiability problem, XSAT,
where it asks if one can find a satisfying assignment such that exactly one of the literals in each clause is assigned 
the value ``1" and all other literals in the same clause are assigned ``0". Another variant that has been heavily
studied is the restriction of the number of literals allowed in each clause. In both SAT and XSAT, one allows arbitrary
number of literals to be present in each clause. If we restrict the number of literals to be at most $k$ in each clause, then
the above problems are now known as $k$SAT and X$k$SAT respectively. The most famous of these variants are 
3SAT and X3SAT. All the mentioned problems, SAT, 3SAT, XSAT and X3SAT
are known to be NP-complete \cite{Cook71,Karp72,Sch78}. 

Apart from decision problems and optimization problems, one can also work on counting the number of different
models that solves the decision problem. For example, we can count
the number of different satisfying assignments
that solves SAT, and this is known as \#SAT. The problem \#3SAT, \#XSAT and \#X3SAT are defined similarly.
Counting problems seem much harder than their decision counterparts. One may use the output
of a counting algorithm to solve the decision problem. Another convincing example can be seen in that 2SAT 
is known to be in P \cite{Krom67} but \#2SAT is \#P-complete \cite{Valiant79}. In fact, \#SAT, \#3SAT,
\#X3SAT and \#XSAT are all known to be in \#P-complete \cite{Valiant79,Valiant79_2}. 
The problem of model counting has found wide applications in the field of AI such as the use of 
inference in Bayesian belief networks or probabilistic inference  
\cite{Roth96,STPH05}. In this paper, we will focus on the \#X3SAT problem.

Let $n$ denote the number of variables in the formula. Algorithms to solve \#XSAT have seen numerous
improvements \cite{Dah02,DJB05,Por05,Wah07} over the years.
To date, the fastest \#XSAT algorithm runs in
$O(1.1995^n)$ time \cite{ZJWJ14}.
Of course, to solve the \#X3SAT problem, one can rely on
any of the mentioned algorithm that solves \#XSAT to solve them directly. However, it is possible to exploit the structure of
X3SAT and hence solve \#X3SAT in a much faster manner.  Dahll\"of, Jonsson and Beigel gave an \#X3SAT
algorithm in $O(1.1487^n)$ time \cite{DJB05}.

In this paper, we propose a faster and simpler algorithm to solve the \#X3SAT problem in $O(1.1120^n)$ time. 
The novelty here lies in the use of a result by Monien and Preis \cite{MP06} to help us to deal
with a specific case. Also using a different way to analyze our algorithm allows us to tighten the analysis further.
\section{Preliminaries}

\noindent
In this section, we will introduce some common definition needed by the algorithm and also the techniques needed
to understand the analysis of the algorithm. The main design of our algorithm is a 
Davis Putnam Logemann Loveland (DPLL) \cite{DLL62,DP60} style algorithm, or also known as the branch and bound algorithm. 
Such algorithms are recursive in nature and have two kinds of rules associated with them: Simplification and Branching rules.
Simplification rules help us to simplify a problem instance. Branching rules on the other hand,
help us to solve a problem instance by recursively solving smaller instances of the problem. To illustrate the execution of 
the DPLL algorithm, a search tree is commonly used. We assign the root node of the search tree as the original problem.
The subsequent child nodes are assigned whenever we invoke a branching rule. For more information, one may refer to
\cite{FK10}. 

Let $\mu$ denote our parameter of complexity.
To analyse the running time of the DPLL algorithm, one in fact just needs to bound the number of leaves generated
in the search tree. This is due to the fact that the complexity of such algorithm is proportional to the number of leaves, 
modulo polynomial factors, i.e., $O(poly(|\varphi|,\mu) \times \text{number of leaves in the search tree})$ = \\
$O^*(\text{number of leaves in the search tree})$, where the function $poly(|\varphi|,\mu)$ is some polynomial based on
$|\varphi|$ and $\mu$, while $O^*(g(\mu))$ is the class of all functions $f$ bounded by some polynomial $p(\cdot)$ times $g(\mu)$.

Then we let $T(\mu)$ denote the maximum number of leaf nodes generated by the algorithm when we have $\mu$ as
the parameter for the input problem.
Since the search tree is only generated by applying a branching rule, 
it suffices to consider the number of leaf nodes generated by that rule (as simplification rules take only polynomial time).
To do this, we employ techniques in \cite{Kul99}. Suppose a branching rule has $r \geq 2$ children, with
$t_1,t_2 ,\ldots,t_r$ number of variables eliminated for these children.
Then, any function $T(\mu)$ which satisfies $T(\mu) \geq T(\mu-t_1) + T(\mu-t_2) + \ldots T(\mu-t_r)$, with appropriate
base cases, would satisfy the bounds for the branching rule. To solve the above linear recurrence, 
one can model this as $x^{-t_1} + x^{-t_2} + \ldots + x^{-t_r} = 1$. Let $\beta$ be the root
of this recurrence, where $\beta \geq 1$. Then any $T(\mu) \geq \beta^\mu$ would satisfy the recurrence for
this branching rule. In addition, we denote the branching factor $\tau(t_1,t_2,\ldots,t_r)$
as $\beta$. Tuple $(t_1,t_2,\ldots,t_r)$ is also known as the
branching vector\cite{FK10}. 
If there are $k$ branching rules in the DPLL algorithm, then the overall
complexity of the algorithm can be seen as the largest branching factor among all $k$ branching rules;
i.e. $c=max\{\beta_1,\beta_2,\ldots,\beta_k\}$, and therefore the time complexity of the
algorithm is bounded above by $O^*(c^\mu)$.

We will introduce some known results about branching factors. If $k < k'$, then we have that
$\tau(k',j) < \tau(k,j)$, for all positive $k,j$. In other words, comparing two branching factors,
if one eliminates more variable, then this will result in a a smaller branching factor. Suppose that 
$i+j = 2\alpha$, for some $\alpha$, then $\tau(\alpha,\alpha) \leq \tau(i,j)$. In other words, 
a more balanced tree will give a smaller branching factor. 

Finally, suppose that we have a branching vector of $(u,v)$ for some branching rule. Suppose
that for the first branch, we immediately do a follow up branching to get a branching vector 
of $(w,x)$, then we can apply branching vector addition to get a combined branching 
vector of $(u+w,u+x,v)$. This technique can sometimes help us to bring down the 
overall complexity of the algorithm further.

Finally, the correctness of DPLL algorithms usually follows from the fact that all cases
have been covered. We now give a few definitions before moving onto the actual algorithm.
We fix a formula $\varphi$ :

\begin{defn}
Two clauses are called neighbours if they share at least a common variable. Two variables are called
neighbours if they appear in some clause together. We say that a clause $C$ is a degree $k$ clause if
$C$ has $k$ neighbours. Finally, a variable is a singleton if it appears only once in $\varphi$.
\end{defn}

\noindent
Suppose we have clauses $C_1 = (x \vee y \vee z)$, $C_2 = (x \vee a \vee b)$
and $C_3 = (y \vee a \vee c)$. Then $C_1$ is a neighbour to $C_2$ and $C_3$.
In addition, all three are degree 2 clauses. 
Variables $a,b,y,z$ are neighbours of $x$, while $b,c,z$ are singletons.

\begin{defn}\label{defn-link}
We say that two variables, $x$ and $y$, are linked when we can deduce either $x=y$ or $x=\bar{y}$. When this happens,
we can proceed to remove one of the linked variable, either $x$ or $y$, by replacing it with the other.
\end{defn}

\noindent
For example, in clause $(0 \vee x \vee y)$, we know that $x=\bar{y}$ to satisfy it.
Thus, we can link $x$ with $\bar{y}$ and remove one of the variables, say $y$.

\begin{defn}
We denote the formula $\varphi[x=1]$ obtained from $\varphi$ by
assigning a value of $1$ to the literal $x$. We denote the formula $\varphi[x=y]$ as obtained from
$\varphi$ by substituting all instances of $x$ by $y$. Similarly, let $\delta$ be a subclause. We denote
$\varphi[\delta=0]$ as obtained from $\varphi$ by substituting all literals in $\delta$ to 0.
\end{defn}

\noindent
Suppose we have $\varphi=(x \vee y \vee z)$. Then if we assign $x=1$, then $\varphi[x=1]$ gives us 
$(1 \vee y \vee z)$. On the other hand, if we have $\varphi[y=x]$, then we have 
$(x \vee x \vee z)$. If $\delta=(y \vee z)$, then $\varphi[\delta=0]$ gives us
$(x \vee 0 \vee 0)$.

\begin{defn}
A sequence of degree 2 clauses $C_1,C_2,\ldots,C_k$, $k\geq1$ is called a chain if for $2\leq j \leq k-1$,
we have $C_j$ is a neighbour to both $C_{j-1}$ and $C_{j+1}$. Given any two clauses $C_e$ and $C_f$
that are at least degree 3, we say that they are connected via a chain if we have a chain $C_1,C_2,\ldots,C_k$ such that $C_1$ is
a neighbour of $C_e$ (respectively $C_f$) and $C_k$ is a neighbour of $C_f$ (respectively $C_e$). Moreover, if we have a 
chain of degree 2 clauses $C_1,C_2,\ldots,C_k,C_1$, then we call this a cycle.
\end{defn}

\noindent
Suppose we have the following degree 3 clauses : $(a \vee b \vee c)$ and $(s \vee t \vee u)$, and the
following chain : $(c \vee d \vee e)$, $(e \vee f \vee g)$, $\ldots$,
$(q \vee r \vee s)$. Then note that
the degree 3 clause $(a \vee b \vee c)$ is a neighbour to $(c \vee d \vee e)$ and $(s \vee t \vee u)$
is a neighbour to $(q \vee r \vee s)$. Therefore, we say that $(a \vee b \vee c)$ and $(s \vee t \vee u)$
are connected via a chain. 
\footnote{The definition of chains and cycles will be mainly used in Section~\ref{sec-line17} and Section~\ref{sec-line18}.} 

\begin{defn}
A path $x_1,x_2,\ldots,x_i$ is a sequence of variables such that for each $j \in \{1,\ldots,i-1\}$, the variables 
$x_j$ and $x_{j+1}$ are neighbours. A component is a maximal set of clauses such that
any two variables, found in any clauses in the set has a path between each other. 
A formula is connected if any two variables
have a path between each other. Else we say that the formula is disconnected, and consists of
$k\geq2$ components.
\end{defn}

\noindent
For example, let $\varphi=(x \vee y \vee z) \land (x \vee a \vee b) \land (e \vee c \vee d) \land (e \vee f \vee g)$.
Then $\varphi$ is disconnected and is made up of two components, since $x$ has no path to $e$, while variables in the set
$\{(x \vee y \vee z), (x \vee a \vee b)\}$ have a path to each other. Similarly, for $\{(e \vee c \vee d),(e \vee f \vee g)\}$. 
Therefore, $\{(x \vee y \vee z), (x \vee a \vee b)\}$ and $\{(e \vee c \vee d),(e \vee f \vee g)\}$ are two components.

\begin{defn}
Let $I$ be a set of variables of a fixed size.
We say that $I$ is semi-isolated if there exists an $s \in I$
such that in
any clause involving variables not in $I$, only $s$ from $I$ may appear.
\end{defn}

\noindent
For example consider the set $I=\{x,y,z,a,b\}$ and the clauses 
$(x \vee y \vee z)$, $(x \vee a \vee b)$, $(b \vee c \vee d)$, $(c \vee d \vee
e)$.
Since $b$ is the only variable in $I$ that appears in clauses
involving variables not in $I$, $I$ is semi-isolated. 

\begin{defn}
Suppose $G=(V,E)$ is a simple undirected graph. A {\em balanced bisection} is
a mapping $\pi:V \to \{0,1\}$ such that, 
for $V_i=\{v: \pi(v)=i\}$, $|V_0|$ and $|V_1|$ differ by at most one.
Let $cut(\pi)=|\{(v,w): (v,w) \in E, v \in V_0, w \in V_1\}|$. The bisection
width of $G$ is the smallest $cut(\cdot)$ that can be obtained for a balanced
bisection.
\end{defn}

\noindent

\begin{thm}[see Monien and Preis \cite{MP06}]\label{MPthm}
For any $\varepsilon>0$, there is a value $n(\varepsilon)$ such that the bisection width of any $3$-regular graph
$G=(V,E)$ with $|V|>n(\varepsilon)$ is at most $(\frac{1}{6} + \varepsilon)|V|$.
This bisection can be found in polynomial time.
\end{thm}

\noindent
The above result extends to all graphs $G$ with maximum degree of $3$ \cite{GS17}.

\section{Algorithm}

\noindent
Our algorithm takes in a total of 4 parameters : a formula $\varphi$, a cardinality vector $\vec{c}$,
two sets $L$ and $R$.

The second parameter, a cardinality vector $\vec{c}$,
maps literals to $\mathbb{N}$. 
The idea behind introducing this cardinality
vector $\vec{c}$ is to help us to keep track of the number of models while applying simplification
and branching rules. At the start, $\vec{c}(l)=1$ for all literals in $\varphi$ and will be updated along the way
whenever we link variables together or when we remove singletons. Since linking of variables is a common
operation,  we introduce a function to help us perform this procedure. 
The function $Link(.)$, takes as inputs the cardinality vector and two literals 
involving different variables to link
them \footnote{As seen in Definition~\ref{defn-link}.}. It updates the information of 
the eliminated variable ($y$) onto the surviving variable
($x$) and after which, drops the entries of eliminated variable 
($y$ and $\bar{y}$) in the cardinality vector $\vec{c}$. 
When we link $x$ and $y$ as $x=y$ (respectively, $x=\bar{y}$), 
then we call the function 
$Link(\vec{c},x,y)$ (respectively, $Link(\vec{c},x,\bar{y})$). 
We also use a function $MonienPreis(.)$ to give us partition
based on Theorem~\ref{MPthm}.

\medskip
\noindent
\textbf{Function: $Link(.)$} \\
Input : A Cardinality Vector $\vec{c}$, literal $x$, literal $y$ \\
Output : An updated Cardinality Vector $\vec{c'}$ 
\begin{itemize}
\item Update $\vec{c}(x) = \vec{c}(x) \times \vec{c}(y)$, and $\vec{c}(\bar{x}) = \vec{c}(\bar{x}) \times \vec{c}(\bar{y})$.
After which, drop entries of $y$ and $\bar{y}$ from $\vec{c}$ and update it as $\vec{c'}$. Finally, return $\vec{c'}$ 
\end{itemize}

\noindent
\textbf{Function : $MonienPreis(.)$} \\
Input : A graph $G_{\varphi}$ with maximum degree 3 \\
Output : $L$ and $R$, the left and right partitions of minimum bisection width 

\medskip
\noindent
For the third and fourth parameter, we have the sets of clauses $L$ and $R$. $L$ and $R$ will be used to store
partitions of clauses after calling $MonienPreis(.)$, based on the minimum bisection width.  
Initially, $L$ and $R$ are empty sets and will continue to be until we first come to Line 17 of the algorithm.
\footnote{More details about their role will be given in Section~\ref{sec-line17}.}

We call our algorithm $CountX3SAT(\cdot)$. Whenever a literal $l$ is assigned a constant value, we drop
both the entries $l$ and $\bar{l}$ from the cardinality vector and multiply the returning recursive call 
by $\vec{c}(l)$ if $l=1$, or $\vec{c}(\bar{l})$
if $\bar{l}=1$. In each recursive call, we ensure that the cardinality vector is updated to contain only entries
where variables in the remaining formula have yet to be assigned a constant value. By doing so, we guarantee the 
following invariant : For any given $\varphi$, let 
$S_{\varphi} = \{ h : h \text{ is an exact-satisfiable assignment for } 
\varphi\}$. Now
for any given $\varphi$ and a cardinality vector $\vec{c}$, the output of 
$CountX3SAT(\varphi,\vec{c},L,R)$
is given as $\sum_{h \in S_{\varphi}} \prod_{l : l \text{ is assigned true in } h} \vec{c}(l)$.
Initial call to our algorithm
would be $CountX3SAT(\varphi,\vec{c},\emptyset,\emptyset)$,
where the cardinality vector $\vec{c}$ has $\vec{c}(l)=1$ for all 
literals at the start. 
The correctness of the algorithm follows from the fact that
each step will maintain the invariant that
$CountX3SAT(\varphi,\vec{c},L,R)$ returns
$\sum_{h \in S_{\varphi}} \prod_{l : l \text{ is assigned true in }
h} \vec{c}(l)$, where if $\varphi$ is not exactly satisfiable, it
returns $0$. Note that in the algorithm below 
possibilities considered are exhaustive.

\medskip\noindent
Algorithm : CountX3SAT(.)  \\
Input : A formula $\varphi$, a cardinality vector $\vec{c}$, a set $L$, a set $R$ \\
Output : $\sum_{h \in S_{\varphi}} \prod_{l : l \text{ is assigned true in } h} \vec{c}(l)$ 
\begin{algorithmic}[1]
\STATE If any clause is not exact satisfiable (by analyzing
this clause itself) then return 0. If all clauses consist of constants
evaluating to $1$ or no clause is left then return $1$.
\STATE If there is a clause $(1 \vee \delta)$, then let $\vec{c'}$ be the new cardinality vector
by dropping the entries of the variables in $\delta$. Drop this clause from $\varphi$. \\
Return $CountX3SAT(\varphi[\delta=0],\vec{c'},L,R) \times \prod_{i \text{ is a literal in } \delta} \vec{c}(\bar{i})$
\STATE If there is a clause $C=(0 \vee \delta)$, then update $C=\delta$ in $\varphi$. \\
Return $CountX3SAT(\varphi,\vec{c},L,R)$.
\STATE If there is a single literal $x$ in a clause,
then let $\vec{c'}$ be the new cardinality
vector by dropping the entries $x$ and $\bar{x}$ from $\vec{c}$. \\
Return $CountX3SAT(\varphi[x=1],\vec{c'},L,R) \times \vec{c}(x)$.
\STATE If there is a 2-literal clause $(x \vee y)$, for some literals $x$ and $y$ with $x\neq y$ and $x\neq \bar{y}$, then 
$\vec{c'} = Link(\vec{c},x,\bar{y})$. 
Return $CountX3SAT(\varphi[y=\bar{x}],\vec{c'},L,R)$.
\STATE If there is a clause $(x \vee \bar{x})$, for some variable $x$. Check if $x$ appears in other clauses.
If yes, then drop this clause from $\varphi$ and return $CountX3SAT(\varphi,\vec{c},L,R)$. 
If no, then let $\vec{c'}$ be the new cardinality vector 
by dropping $x$ and $\bar{x}$. Drop this clause from $\varphi$ and return 
$CountX3SAT(\varphi,\vec{c'},L,R) \times (\vec{c}(x) + \vec{c}(\bar{x}))$ .
\STATE If there are $k\geq2$ components in $\varphi$ and there are no edges between $L$ and $R$, 
then let $\varphi_1,\ldots,\varphi_k$
be the $k$ components of $\varphi$. Let $\vec{c}_i$ be the cardinality vector for $\varphi_i$ by
only keeping the entries of the literals involving
variables appearing in $\varphi_i$, and dropping the rest. Let $L=R=\emptyset$.
Return $CountX3SAT(\varphi_1,\vec{c}_1,L,R) \times \ldots \times
CountX3SAT(\varphi_k,\vec{c}_k,L,R)$.
\STATE If there exists a clause $(x \vee x \vee y)$, for some literals $x$ and $y$, 
then let $\vec{c'}$ be the new cardinality vector by dropping the entries $x$ and $\bar{x}$ from $\vec{c}$. \\
Return $CountX3SAT(\varphi[x=0],\vec{c'},L,R) \times \vec{c}(\bar{x})$
\STATE If there is a clause $(x \vee \bar{x} \vee y)$, then let $\vec{c'}$ be the new cardinality
vector by removing the entries $y$ and $\bar{y}$. Return 
$CountX3SAT(\varphi[y=0],\vec{c'},L,R) \times \vec{c}(\bar{y})$
\STATE If there exists a clause containing two singletons $x$ and $y$, then update $\vec{c}$ as : \\
$\vec{c}(x) = \vec{c}(x)\times \vec{c}(\bar{y}) + \vec{c}(\bar{x})\times \vec{c}(y)$, 
$\vec{c}(\bar{x}) = \vec{c}(\bar{x})\times \vec{c}(\bar{y})$. \\
Let $\vec{c'}$ be the new cardinality vector by dropping the entries $y$ and $\bar{y}$ from $\vec{c}$.
Drop $y$ from $\varphi$. Return $CountX3SAT(\varphi,\vec{c'},L,R)$.
\STATE There are two clauses $(x \vee y \vee z)$ and $(x \vee y \vee w)$, for some literals
$x,y,z$ and $w$. Then in this case, let $\vec{c'} = Link(\vec{c},z,w)$. Drop one of the clauses.
Return $CountX3SAT(\varphi[w=z],\vec{c'},L,R)$.
\STATE There are two clauses $(x \vee y \vee z)$ and $(x \vee \bar{y} \vee w)$, for some
literals $x,y,z$ and $w$. Then let $\vec{c'}$ be the new cardinality vector by dropping entries of $x$
and $\bar{x}$. Return $CountX3SAT(\varphi[x=0],\vec{c'},L,R) \times \vec{c}(\bar{x}).$
\STATE There are two clauses $(x \vee y \vee z)$ and $(\bar{x} \vee \bar{y} \vee w)$, 
for some literals $x,y,z$ and $w$. Then $\vec{c'} = Link(\vec{c},x,\bar{y})$.
Return $CountX3SAT(\varphi[y=\bar{x}],\vec{c'},L,R)$. 
\STATE If there exists a semi-isolated set $I$, with $3 \leq |I| \leq20$, then let $x$ be the variable
appearing in further clauses with variables not in $I$. Let $\vec{c'}$ be the new cardinality vector by
updating the entries of $x$ and $\bar{x}$, dropping of entries of variables in $I-\{x\}$.  Drop 
all the entries of $I-\{x\}$ from $\varphi$. Return $CountX3SAT(\varphi,\vec{c'},L,R)$. 
\footnote{More details on the updating of $\vec{c'}$ below in this section.} 
\STATE 
This rule is not analyzed for all cases, but only specific
cases as mentioned in Sections~\ref{sec-line15} and~\ref{sec-line16}
(more specifically this applies only when some variable appears 
in at least 3 clauses).
If there exists a variable $x$ such that branching $x=1$ and $x=0$ 
allows us to either remove 
at least 7 variables on both branches, or at least 8 on one and 6 on the other,
or at least 9 on one and 5 on the other,
then branch $x$. 
Let $\vec{c'}$ be the new cardinality vector by dropping the entries 
$x$ and $\bar{x}$.
Return $CountX3SAT(\varphi[x=1],\vec{c'},L,R) \times \vec{c}(x) 
+ CountX3SAT(\varphi[x=0],\vec{c'},L,R) \times \vec{c}(\bar{x})$.
\footnote{More details on this branching rule is given in Section~4.}
\STATE If there exists a variable $x$ appearing at least 3 times, then let
$\vec{c'}$ be the new cardinality vector by dropping the entries $x$ and $\bar{x}$.
Return $CountX3SAT(\varphi[x \linebreak[3] =1],\vec{c'},L,R) \times \vec{c}(x) + CountX3SAT(\varphi[x=0],\vec{c'},L,R) \times \vec{c}(\bar{x})\  {}^7$.
\STATE If there is a degree 3 clause in $\varphi$, then check if $\exists$ an edge between $L$ and $R$. If no,
then construct $G_{\varphi}$ and let $(L',R') \gets MonienPreis(G_{\varphi})$. 
Then return $CountX3SAT(\varphi,\vec{c},L',R').$ If 
$\exists$ an edge between $L$ and $R$, 
apply only the simplification rules 
(if any) as stated in Section~\ref{sec-line17}. Choose an edge $e$ between $L$ and $R$. Then branch the variable $x_e$
represented by $e$. 
Let the cardinality vector $\vec{c'}$ be the new cardinality vector by dropping off entries $x_e$ and $\bar{x}_e$.
Return $CountX3SAT(\varphi[x_e=1],\vec{c'},L,R) \times \vec{c}(x_e) + 
CountX3SAT(\varphi[x_e=0],\vec{c'},L,R) \times \vec{c}(\bar{x}_e)\ {}^5$.
\STATE If every clause in the formula is degree 2, choose any variable $x$ and we branch $x=1$ and $x=0$.
 Let $\vec{c'}$ be the new cardinality vector by dropping the entries $x$ and $\bar{x}$.
Return $CountX3SAT(\varphi[x=1],\vec{c'},L,R) \times \vec{c}(x) + CountX3SAT(\varphi[x=0],\vec{c'},L,R) \times \vec{c}(\bar{x})\ {}^7$.
\end{algorithmic}

\noindent
Note that every line in the algorithm has descending priority; Line 1 has higher priority than Line 2,
Line 2 than Line 3 etc. 

Line 1 of the algorithm is our stopping condition. If any clause is not exact satisfiable, immediately return 0.
When no variables are left, then check if every clause has been dropped off. If yes, then return 1, else 0. 

Line 2 of the algorithm deals with any clause that contains a constant $1$. In this case, all the other literals in the clause must be 
assigned $0$ and we can safely drop off this clause after that. 
Line 3 deals with any clause with a constant $0$ in it. We can then safely drop the constant $0$ from the clause.
Line 4 deals with single-literal clauses. This literal must be assigned $1$. Line 5 deals with two literal clauses when the two literals
involve two different variables. Line 6 deals with two literal clauses when they come from the same variable, say $x$. 
Now if $x$ does not appear elsewhere, then either $x=1$ or $x=0$ will satisfy this clause. 
Thus as done in Line 6, multiplying 
$CountX3SAT(\varphi,\vec{c'},L,R)$ by the sum of 
$(\vec{c}(x) + \vec{c}(\bar{x}))$ would give us the correct value.
Regardless of whether $x$ appears elsewhere or not, drop this clause.

After Line 6, we know that all clauses are of length 3. In Line 7, if the formula is disconnected,
then we deal with each components separately. Line 7 has some relation with Line 17. If the algorithm is not currently
processing Line 17, then basically we just call the algorithm on different components. The explicit relationship between Line 7 and
Line 17 will be given in Section~\ref{sec-line17}.
In Line 8, we deal with a literal that appears twice in a clause. 
Then we can assign that literal as $0$. In Line 9, we have a literal and its negation appearing
in the same clause, then we assign the last literal to be $0$. In Line 10, we deal with clauses having two singletons and we need
to update the cardinality vector $\vec{c}$ before we are allowed to remove one. Suppose we have two singletons $x$ and $y$
and we wish to remove say $y$, then we need to update the entries of $\vec{c}(x)$ and $\vec{c}(\bar{x})$ 
to retain the information of $\vec{c}(y)$ and $\vec{c}(\bar{y})$.
Note that in the updated $x$, when $x=0$, this means that both the original $x$ and $y$ are 0.
On the other hand, when we have $x=1$ in the updated $x$, this means that we can either have $x=1$ in the original $x$,
or $y=1$. Thus, this gives us the following update : 
$\vec{c}(x) = \vec{c}(x)\times \vec{c}(\bar{y}) + \vec{c}(\bar{x})\times \vec{c}(y)$ when $x$ is assigned ``1", and
$\vec{c}(\bar{x}) = \vec{c}(\bar{x})\times \vec{c}(\bar{y})$ when $x$ is assigned ``0". After which, we can then
safely remove the entries of $y$ and $\bar{y}$ from the cardinality vector $\vec{c}$.

In Lines 11, 12 and 13, we deal with two overlapping variables (in different permutation) between any two clauses. After which,
any two clauses can only have at most only 1 overlapping variable between them. In Line 14, we deal with semi-isolated sets
$I$ such that we can remove all but one of its variable. In Line 15, if we can find a variable $x$ such that by branching it,
we can remove that amount of variables as stated, then we proceed to do so. 
The goal of introducing Line 14 and Line 15 is to help us out for Line 16,
where we deal with variables that appear at least 3 times. Their relationship will be made clearer in the later sections.
After which, all variables will appear at most 2 times and each clause must have at most degree 3.
In Line 17, the remaining formula must consist of clauses of degree 2 and 3. Then we construct
a graph $G_{\varphi}$, apply $MonienPreis(.)$ to it and choose a variable to branch, followed by applying simplification rules.
We'll continue doing so until no degree 3 clauses exist.
Lastly in Line 18, the formula will only consist of degree 2 clauses, and we will select any variable and branch $x=1$ and $x=0$.
Hence, we have covered all cases in the algorithm.

Now, we give the details of Line 14. As $I$ 
is semi-isolated, let $x$ be the variable in $I$, such that
$x$ appears in further clauses containing variables not in $I$. 
Note that when $x=1$ or when $x=0$, the formula
becomes disconnected and clauses involving
$I-\{x\}$ become a component of constant size. Therefore, we can use
brute force (requiring constant time), 
to check which assignments to the $|I|-1$ variables satisfy the
clauses involving variables from $I$, and then correspondingly update 
$\vec{c}(x)$ and $\vec{c}(\bar{x})$,
and drop all variables in $I-\{x\}$ from $\varphi$. 
We call such a process {\em contraction} of $I$ into $x$.
Details given below.

\medskip
\noindent\textbf{Updating of Cardinality vector in Line 14 
(Contracting variables).}
Let $S$ be the set of clauses which involve only variables in $I$.
$\delta$ below denotes assignments to variables in $I-\{x\}$.
For $i \in \{0,1\}$, let 
\begin{quote}
$Z_i=\{\delta:$ all clauses in $S$ are satisfied when variables in
$I$ are set according to $\delta$ and $x=i\}$.
\end{quote}
The following formulas update the cardinality vector for coordinate
$x$ and $\bar{x}$, by considering the different possibilities
of $\delta$ which make the clauses in $S$ satisfiable.
This is done by summing over all such $\delta$ in $Z_i$ 
(for $i=x=0$ and $i=x=1$), the multiplicative factor formed by
considering the cardinality vector values at the corresponding
true literals in $\delta$. Here the literals $\ell$ in
the formula range over literals involving the variables in $I-\{x\}$.

Let $\vec{c}(x) = \vec{c}(x) \times \sum_{\delta \in Z_1} 
\prod_{\ell \text{ is true in } \delta} \vec{c}(\ell)$.

Let $\vec{c}(\bar{x}) = \vec{c}(\bar{x}) \times \sum_{\delta \in Z_0} 
\prod_{\ell \text{ is true in } \delta} \vec{c}(\ell)$.

\section{Analysis of the Branching Rules of the Algorithm}

\noindent
Note that Lines 1 to 14 are simplification rules and Lines 15 to
18 are branching rules. For Line 7, note that since the time 
of our algorithm is running in $O^*(c^n)$, for some $c$, then
calling our algorithm onto different components will still give us $O^*(c^n)$.
Therefore, we will analyse Lines 15 to 18 of the algorithm.

\subsection{Line 15 of the algorithm}\label{sec-line15}

\noindent
The goal of introducing Lines 14 and 15 is to ultimately help us to simplify our cases
when we deal with Line 16 of the algorithm. In Line 16, there can be some ugly
overlapping cases which we don't have to worry after adding 
Lines 14 and 15 in the algorithm.
The cases we are interested in are as follows.

(A) There exists a variable which appears in at least four clauses.

Suppose the variable is $x_0$, and the four clauses it appears in are 
$(x_0' \vee x_1 \vee x_2)$,
$(x_0'' \vee x_3 \vee x_4)$,
$(x_0''' \vee x_5 \vee x_6)$,
$(x_0'''' \vee x_7 \vee x_8)$, 
where $x_0', x_0'', x_0''',x_0''''$ are either $x_0$ or $\bar{x_0}$.
Note that $x_0,x_1,x_2,\ldots,x_8$ are literals involving different variables
(by Lines 8,9,11,12,13).
Note that setting literal $x_0'$ to $1$ will correspondingly set both 
$x_1$ and $x_2$ to $0$; when $x_0'$ is set to $0$ correspondingly 
$x_1$ and $\bar{x_2}$ get linked. Similarly, when we set
$x_0'', x_0''', x_0''''$.
Thus, setting $x_0$ to $1$ or $0$ will give us removal of
$i$ variables on one setting and $12-i$ variables on the other setting,
where $4 \leq i \leq 8$.
Thus, including $x_0$, this gives us,
in the worst case, a branching factor of $\tau(9,5)$.

(B) There exists a variable which appears in exactly three clauses.

Suppose $x_0$ is a variable appearing in the three clauses
$(x_0' \vee x_1 \vee x_2)$,
$(x_0'' \vee x_3 \vee x_4)$,
$(x_0''' \vee x_5 \vee x_6)$
where $x_0', x_0'', x_0'''$ are either $x_0$ or $\bar{x_0}$.
Note that $x_0,x_1,x_2,\ldots,x_6$ are literals involving 
different variables.
Let $I=\{x_0,v_1,v_2,\ldots,v_6\}$, where
$v_i$ is the variable for the literal $x_i$. 

(B.1) If $I$ is
semi-isolated, or $I \cup \{u\}$ is semi-isolated
for some variable $u$, then Line 14 takes care of this.

(B.2) If there are two other variables $u,w$ which 
may appear in any clause involving variables from $I$,
then we can branch on one of the variables $u$ and then
do contraction as in Line 14 for $I \cup \{w\}$ to $w$.
Thus, we will have a branching factor of at least 
$\tau(8,8)$.

(B.3) If there are at most two clauses $C1$ and $C2$ which
involve variables from $I$ and from outside $I$ and these two
together involve at least three variables from outside $I$,
then consider the following cases.

Case 1: If both $C1$ and $C2$ have two variables from
outside $I$. 
Then, let $C1$ have literal $x_i'$
and $C2$ have literal $x_j'$, where $x_i'$ is either $x_i$ or
$\bar{x_i}$ and $x_j'$ is either $x_j$ or $\bar{x_j}$, and
$i,j\in\{0,1,\ldots,6\}$.
Now, one can branch on literal $x_i'$ being $1$ or $0$.
In both cases, we can contract the remaining variables
of $I$ into $x_j$ (using Line 14). Including
the two literals set to $0$ in $C1$ when $x_i'$ is $1$,
we get branching factor of $\tau(8,6)$.

Case 2: $C1$ and $C2$ together have three variables from outside
$I$. Without loss of generality assume
$C1$ has one variable from outside $I$ and $C2$
has two variables from outside $I$.
Then let $C1$ have literal $y$ which is outside $I$ 
and $C2$ have literal $x_j'$, 
where $x_j'$ is either $x_j$ or $\bar{x_j}$.
Now, one can branch on literal $y$ being $1$ or $0$.
In both cases, we can contract the variables
of $I$ into $x_j$ (using Line 14). 
Including the literal $y$ we get branching factor of $\tau(7,7)$.

(B.4) Case 2.3 and Case 2.4 in Lemma~\ref{lem-line16} for Line 16.  

\begin{lem}
Branching the variable in Line 15 takes $O(1.1074^n)$ time.
(The worst branching factor is $\tau(9,5)$.)
\end{lem}

\subsection{Line 16 of the algorithm}\label{sec-line16}

\noindent
In this case, we deal with variables that appear exactly 3 times. 

\begin{lem}\label{lem-line16}
The time complexity of branching variables appearing 3 times is $O(1.1120^n)$.
\end{lem}
\begin{proof}
Suppose $x_0$ appears three times. Then we let the clauses that 
$x_0$ appear in be 
$(x_0' \vee x_1 \vee x_2)$, $(x_0'' \vee x_3 \vee x_4)$, 
$(x_0''' \vee x_5 \vee x_6)$,
where the primed versions of $x_0$ denote either $x_0$ or $\bar{x_0}$. 

Let $I=\{x_0,v_1,\ldots,v_6\}$, where $v_i$ is the variable in the
literal $x_i$. 

Note that when $x_0'$ is set to $1$, then $x_1$ and $x_2$ are also
set to $0$. When $x_0'$ is set to $0$ then $x_1$ and $x_2$ get linked.
Similarly, for setting of $x_0''$ and $x_0'''$. Thus,
setting of $x_0$ to $1$ or $0$ allows us to remove
$i$ variables and $9-i$ variables respectively among $v_1, \ldots, v_6$,
where $3 \leq i \leq 6$ (the worst case for us thus happens
with removal of $3$ variables on one side and $6$ on the other).
We will show how to remove three further variables outside $I$ in
the following cases (these may fall on either side of setting
of $x_0$ to $1$ or $0$ above). 
Including $x_0$, we get the worst case branching factor of $\tau(10,4)$.

Let the variables outside $I$ be called outside variables for this
proof. Let a clause involving both variables from $I$ and
outside $I$ be called a mixed clause.
By Line 14 and 15 of the algorithm, 
there are at least 3 mixed clauses,
and at least three outside variables which appear in mixed clauses.

Consider 3 mixed clauses
$C1 = (x_i' \vee a_1 \vee a_2)$, 
$C2 = (x_j' \vee a_3 \vee a_4)$ and $C3 = (x_k' \vee a_5 \vee a_6)$,
where $a_2, a_4, a_6$ are literals involving outside variables,
and $x_i',x_j',x_k'$ are literals involving variables from $I$.

Case 1: It is possible to select the three mixed clauses such
that $a_4$ involves a variable not appearing in $C1$
and $a_6$ involves a variable not appearing in $C1,C2$.

Note that this can always be done when there are at least four outside
variables which appear in some mixed clauses.

In this case, $x_i'$ is set in at least one of the cases of $x_0$ being set
to $1$  or $0$. Similarly for $x_j'$ and $x_k'$.
In the case when $x_i'$ is set, one can either set $a_2$ or
link it to $a_1$. 
In the case when $x_j'$ is set, one can either set $a_4$ or
link it to $a_3$. 
In the case when $x_k'$ is set, one can either set $a_6$ or
link it to $a_5$. Note that the above linkings are not cyclic
as the variable for $a_4$ is different from that of $a_1$ and $a_2$.
and the variable for $a_6$ is different from that of $a_1,a_2,a_3,a_4$.
Thus, in total three outside variables are removed when $x_0$ is
set to $1$ and $0$. 

Case 2: Not Case 1.
Here, the number of outside variables which appear in some mixed
clause is exactly three. 
Choose some mixed clauses $C1, C2, C3$ such that exactly three
outside variables are present in them.
Suppose these variables are $a,b,c$.
Suppose the number of outside variables in C1, C2, C3 is
given by triple $(s_1,s_2,s_3)$ (without loss of generality
assume $s_1 \leq s_2 \leq s_3$).
We assume that the clauses chosen are so as to 
have the earlier case applicable below. That is, if all
three variables $a,b,c$ appear in some mixed clause as only
outside variable, then Case 2.1 is chosen;
Otherwise, if at least 2 mixed clauses involving 2 outside variables are there
and a mixed clause involving only one outside variable is there
then Case 2.2. is chosen. Otherwise, if only one mixed clause involving two
outside variable is there then Case 2.3 is chosen. Else, case 2.4 is chosen.

Case 2.1: $(s_1,s_2,s_3)=(1,1,1)$. This would fall in Case 1,
as all three outside variables are different.

Case 2.2: $(s_1,s_2,s_3)=(1,2,2)$. As two variables cannot
overlap in two different clauses, one can assume
without loss of generality that the outside variables in
$C1$ is $a$ or $b$, in $C2$ are $(a,b)$ and $C3$ are $(b,c)$.
But then this falls in Case 1.

Case 2.3: $(s_1,s_2,s_3)=(1,1,2)$. For this not to fall in Case 1,
we must have the same outside variable in $C1$ and $C2$. 
Suppose $a$ appears in $C1, C2$ and $b,c$ in $C3$. 
Furthermore, to not fall in Case 1, we must have that
all other outside clauses must have $a$ only as the
outside variable (they cannot have both $b,c$ as outside variable,
as overlapping of two variables is not allowed).
Thus, by branching on $a$, and then contracting, using Line 14, 
$I$ to $x_k$, 
will allow us to have a worst case branching factor $\tau(7,7)$.
Thus, this is covered under Line 15.

Case 2.4: $(s_1,s_2,s_3)=(2,2,2)$. Say $a,b$ are the outside
variables in C1, $a,c$ are the outside variables in $C2$ and
$b,c$ are the outside variables in $C3$. Furthermore, no other
mixed clauses are there (as no two clauses can overlap in two
literals).

Case 2.4.1: At least one of $a, b, c$ appears both as positive and negative
literal in $C1, C2, C3$.

Suppose without loss of generality that $a$ appears as positive in
$C1$ and negative in $C2$.
Then, setting $a$ to be $1$, allows us to set $b$ as well as 
contract all of $I$ to $c$ using Line 14.
Setting $a$ to be $0$, allows us to set $c$ as well as contract
all of $I$ to $b$ using Line 14. Thus, we get a worst case
branching factor of $\tau(9,9)$.

Thus, this is covered under Line 15.

Case 2.4.2: None of $a,b,c$ appears both as positive and negative
literal in $C1, C2, C3$. Without loss of generality assume $a,b,c$
all appear as positive literals in $C1, C2, C3$.

When, we set $x_i'=1$, we have that $a=b=0$ and we can contract
rest of $I$ to $c$ using Line 14. This gives us removal of $9$ variables.
When we set $x_i'=0$, we have that $a=\bar{b}$, and thus
$c$ must be $0$ (from $C2$ and $C3$), and thus we can contract 
rest of $I$ into $a$ using Line 14.
Thus we get a worst case branching factor of $\tau(9,9)$.
Thus, this is covered under Line 15.

Therefore, the worst case time complexity is 
$O(\tau(10,4)^n) \subseteq O(1.1120^n)$.
\end{proof}

\subsection{Line 17 of the algorithm}\label{sec-line17}

\noindent
We now deal with degree $3$ clauses. 

\begin{algorithmic}[1]
\setcounter{ALC@line}{16}
\STATE If there is a degree 3 clause in $\varphi$, then check if $\exists$ an edge between $L$ and $R$. If no,
then construct $G_{\varphi}$ and let $(L',R') \gets MonienPreis(G_{\varphi})$. 
Then return $CountX3SAT(\varphi,\vec{c},L',R').$ If 
$\exists$ an edge between $L$ and $R$,
apply only the simplification rules 
(if any) as stated in this section (Section~\ref{sec-line17}).
Choose an edge $e$ between $L$ and $R$. Then branch the variable $x_e$
represented by $e$. 
Let the cardinality vector $\vec{c'}$ be the new cardinality vector by dropping off entries $x_e$ and $\bar{x}_e$.
Return $CountX3SAT(\varphi[x_e=1],\vec{c'},L,R) \times \vec{c}(x_e) + 
CountX3SAT(\varphi[x_e=0],\vec{c'},L,R) \times \vec{c}(\bar{x}_e)$. 
\end{algorithmic}

\noindent
Now, we discuss Line 17 of the algorithm in detail. As long as a degree 3 clause exists in the formula, we repeat this 
process. First, we describe how to construct the graph $G_{\varphi}$.

\medskip
\noindent\textbf{Construction.}
We construct a graph $G_{\varphi} = (V,E)$, 
where $V=\{v_C : C$ is a degree 3 clause in $\varphi\}$.
Given any vertices $v_{C'}$ and $v_{C''}$, we add an edge between them if any of the below conditions occur
on clauses $C'$ and $C''$, where $C'$ and $C''$ are clauses with 3 neighbours :
\begin{enumerate}
\item If a common variable appears in both $C'$ and $C''$
\item $C'$ and $C''$ are connected by a chain of 2-degree clauses.
\end{enumerate}
By construction, the graph $G_{\varphi}$ has maximum degree 3. Let $m_3$ denote the number
of degree 3 clauses in $\varphi$. This gives us $|V|=m_3$.
We can therefore apply the result by Monien and Preis, with the size of the bisection width 
$k \leq m_3(\frac{1}{6}+\varepsilon)$. 

We construct the graph $G_{\varphi}$ when there are no edges between $L$ and $R$,
and then apply $MonienPreis(.)$ to get our new partitions $L'$ and $R'$, which are sets of clauses.
These partitions will remain connected until all edges between them
are removed. In other words, the variables represented by them are branched. Now instead of bruteforcing
all the variables in the bisection width at the same time, we branch them edge by edge.
After each branching,
we apply simplification rules before branching again. 
By our construction, we will not increase the degree of our
clauses or variables (except temporarily due to linking; 
the corresponding clause will then be removed via Line 6).
Therefore, we never need to resort to the earlier 
branching rules (Line 15 and 16) that 
deal with variables appearing at least 3 times again.
In other words, once we come into Line 17, we will 
be repeating this branching rule in a recursive manner
until all degree 3 clauses have been removed. 
Applying the simplification rules
could mean that some variables have been removed directly or via linking, 
or some degree 3 clauses have now been dropped to a degree 2 clause etc. In other words, the clauses in the 
sets $L$ and $R$ have changed. Therefore, we need to update $L$ and $R$ 
correspondingly to reflect these changes before we repeat the branching again.

After branching the last variable between the two partitions, the formula becomes disconnected with 
two components and Line 7 handles this. 
Recall that in Line 7, we gave an additional condition to check for any edges
between $L$ and $R$. During the course of applying simplification rules or branching the variables, 
it could be that additional components can be created before all the edges between $L$ and $R$ have been removed.
Therefore, this condition to check for any edges between the partition is to ensure that Line 7 will not be called prematurely
until all edges have been removed. We will now
give in detail the 
choosing of the variable to branch below. 

\medskip
\noindent\textbf{Choosing of variables to branch.}
Based on the construction earlier, an edge is added if any of the two possibilities mentioned above
happen in the formula.
Let $e$ be an edge in the bisection width.
We choose a specific variable to branch in the different scenarios listed.

\begin{enumerate}
\ifnum\vers=0
\item Case 1 : The edge $e$ represents a variable sitting on two degree 3 clauses. For example we have 
two degree 3 clauses $(r \vee s \vee t)$ and $(t' \vee u \vee v)$, where $t'=t$ or $t'=\bar{t}$, 
and these degree 3 clauses represent the two vertices. 
The edge $e$ is represented by the variable $t$. For such cases, we branch $t$.
\else
\item Case 1 : The edge $e$ represents a variable sitting on two degree 3 clauses. Branch this variable.
\fi
\item Case 2 : 
The edge $e$ represents a chain of 2 degree clauses. 
We alternate the branchings between the variables that appear in a degree 3 clause and a 
degree 2 clause at both ends whenever Case 2 arises for symmetry reasons.
For example, if we have degree 3 clause $(a \vee b \vee c)$ in the left
partition connected to degree 3 clause $(s \vee t \vee u)$ in the right partition
via a chain $(c,d,e), \ldots,(q,r,s)$, and it is left partition end turn, 
then we branch on variable $c$; if it is right partition end turn
then we branch on variable $s$. These branchings will remove the whole chain,
and convert the two degree 3 clauses into degree two or lower clause by compression
as described below.
\end{enumerate}

\ifnum\vers=0
\noindent
We alternate our branchings in Case 2 for symmetry reasons, so that the effect on both sides are
the same and therefore, it suffices to concentrate on only one side for our analysis.
If we were to repeatedly branch from
the same side 
for Case 2, then the number of degree 3
clauses removed in both components may differ significantly. \\
\fi

\noindent\textbf{Compression.}
Suppose $C'$ and $C''$ are two degree 3 clauses connected via a chain
$C_1, C_2, \ldots, C_k$, where $c$ is a common variable between $C'$ and $C_1$,
and $s$ is a common variable between $C''$ and $C_k$.
When $s$ is assigned either a value of 
0 or 1, $C''$ drops to a clause of degree at most 2. $C_k$
becomes a 2-literal clause (in the worst case) and we can link the two remaining
literals in it together and the clause is dropped. 
Therefore, the neighbouring clause $C_{k-1}$ has now become a degree 1 clause. By Line
10 of the algorithm, we can remove 1 singleton and $C_{k-1}$ drops to a 2-literal clause.
Continuing the process of linking, dropping of clause and removing of singletons, 
the degree 3 clause at the end, $C'$, will drop to become a clause of at most degree 2 
when $C_1$ is removed. Therefore, $C'$ and $C''$ will drop to a clause of at most degree 2.

With the Compression method, we now have the following. Let $C$ be a degree 3 
clause. Since $C$ is a degree 3 clause, it has an edge to three other degree 3 clauses, say $E_1,E_2,E_3$.
Choose any edge, say between $E_1$ and $C$. Now this edge can either represent a variable appearing in
both $C$ and $E_1$, or a chain between $E_1$ and $C$ with variables at both ends appearing in $E_1$ and $C$. 
Therefore, assigning a value of 0 or 1 to this 
chosen variable represented by the edge will cause $C$ to drop to a 
clause of degree at most 2. 

\medskip
\ifnum\vers=0
\noindent\textbf{Self-loop.}
Note that such a special case can arise, where a degree 3 clause can be
connected via a degree 2 chain to itself. Let $C=(x \vee y \vee z)$
be a degree 3 clause where $y$ and $z$ appear at the end of a degree 2
chain. We proceed now as follows.

Suppose the 2-chain connecting $C$ to itself is of the form:
$(y' \vee u_1 \vee u_2) (u_2' \vee u_3 \vee u_4) \ldots (u_k'  \vee u_{k+1} 
\vee z')$,
where the primed versions are either negation of or same as the unprimed versions.

We distinguish the cases $y=0$ and $y=1$. In both cases we replace
$y$ in the two clauses by distinct new variables $v,w$. If $y=1$
then the new variables receive in $\vec c$ the values
$\vec{c}(v)=1$, $\vec{c}(w)=1$, $\vec{c}(\bar v)=0$, $\vec{c}(\bar w)=0$
else the new variables receive in $\vec c$ the values
$\vec{c}(v)=0$, $\vec{c}(w)=0$, $\vec{c}(\bar v)=1$, $\vec{c}(\bar w)=1$.

Replacing $y$ by $v$ and $y'$ by $w'$ means, we have the chain:
 $(w' \vee u_1 \vee u_2) (u_2' \vee u_3 \vee u_4) \ldots  (u_k'  \vee u_{k+1} 
  \vee z')$, which connects to the clause $C=(z \vee v \vee x)$
where the left end is now a degree one clause (dead end) 
and $x$ is the only variable which connects the above to the rest of $\varphi$. 

Now we can always contract the deadend degree 1 clause (initially
$(w' \vee u_1 \vee u_2)$) at the left end of above sequence into the variable
connecting it to the rest of $\varphi$ until this variable is $x$
and has in $\vec c$ the entries $\vec c(x:y=b)$,
$\vec c(\bar x:y=b)$ for the case that $y=b$. Now one updates the so
obtained entries of $\vec c$ by the following formula:
\begin{eqnarray*}
   \vec c(x) = \vec c(y) \times \vec c(x:y=1) +
               \vec c(\bar y) \times \vec c(x:y=0); \\
   \vec c(\bar x) = \vec c(y) \times \vec c(\bar x:y=1) +
               \vec c(\bar y) \times \vec c(\bar x:y=0).
\end{eqnarray*}
In the case that after treating the self-loop, $x$ is in a degree 1
clause then one keeps compressing the degree 1 clause at the end
of the chain originally going until $x$ until the whole chain is compressed
into a variable contained in a degree 3 clause, which then becomes
a degree 2 clause. All the variables and clauses which became obsolete,
including $y,v,w$, will be omitted in $\vec c$ and $\varphi$. As this
procedure is the series of at most $n$ compressions of semi-isolated
components consisting of three variables into one variable, the whole
procedure runs in time polynomial in $n$.
\else
\noindent\textbf{Self-loop.}
Note that such a case can arise, where a degree 3 clause can be connected via a
degree 2 chain to itself. The idea to handle this is similar to Line 14 and 
by adopting the idea in Compression. Due to space constraints,
details are omitted.
\fi

\medskip
\noindent
Based on the choice of variables as mentioned above, we now give the time analysis for Line 17 of
the algorithm. Note that the measure of complexity for our branching factors here is $m_3$, the
number of degree 3 clauses.
 
\begin{lem}
The time complexity of dealing of branching variables in the bisection width is $O(1.1092^n)$
\end{lem}

\begin{proofsketch}
For $m_3$, the current number of degree 3 clauses, we have that each variable
in a degree $3$ clause occurs in exactly one further clause and that there are
three variables per clause. Thus $3 m_3 \leq 2 n$ and
$m_3 \leq \frac{2}{3}n$, where $n$ is the current number of variables.
Note that the bisection width has size $k\leq m_3(\frac{1}{6} + \varepsilon)$. 

Once we remove the edges in the bisection width, the two sides (call them left
(L) and right (R)) get disconnected, and thus each component can be solved
independently. Here note that after the removal of all the edges
in the bisection width, we have at most $m_3/2$
degree 3 clauses in each partition.
As we ignore polynomial factors in counting the number of leaves,
it suffices to concentrate on one (say left) partition.
We consider two kinds of reductions:
(i) a degree 3 clause on the left partition is removed or
becomes of degree less than three due to a branching, and
(ii) the degree 3 clauses on the right partition are not part of
the left partition.
The reduction due to (ii) is called bookkeeping reduction because we
spread it out over the removal of all the edges in the bisection width.
Note that after all the edges between $L$ and $R$ have been removed,
$\frac{m_3}{2}$ many clauses are reduced due to the right partition
not being connected to the left partition. As the number
of edges in the bisection width is at most $\frac{m_3}{6}$,
in the worst case, we can count at least
$\frac{m_3}{2} \div \frac{m_3}{6}=3$ degree 3 clauses for each edge 
in the bisection width that we remove.
For the removal of degree 3 clauses in the left partition, we analyze as follows.

Let an edge be given between $L$ and $R$.  We let the degree 3 clause $C=(a \vee b \vee c)$ be on the left partition,
and the degree 3 clause $T=(s \vee t \vee u)$ be on the right partition. Then the edge can be represented 
by $c$, with $s=c$ or $s=\bar{c}$, or the edge is represented by a chain of degree 2 clauses, with 
the ends being $c$ and $s$. We branch the variable $c=1$ and $c=0$.

When $c=0$, $C$ gets dropped to a degree 2 clause. Now this also means that the given
edge gets removed (either directly or via Compression). Counting an additional 3 degree 3 clauses from
the bookkeeping process, we remove a total of 4 degree 3 clauses here. 

When $c=1$, then $a=b=0$. Since $C$ is a degree 3 clause, it is connected to 3 other degree 3 clauses.
Now all 3 degree 3 clauses will either be removed, or will drop to a degree 2 clause 
(again either directly, or via Compression). Hence, this allows us to remove 
$1+3i+(3-i)$ degree 3 clauses, where removing $C$ counts as 1, 
$i$ is the number of neighbours of $C$ in the right partition (bookkeeping)
while $(3-i)$ be the number of neighbours on the left. Since $i\in \{1,2,3\}$, the minimum number of
degree 3 clauses we can remove here happens to be for $i=1$, 
giving us 6 degree 3 clauses for this branch.
This gives us a branching factor of $\tau(6,4)$.

When we branch the variable $s=1$ and $s=0$, 
$C$ gets dropped to a degree 2 clause via Compression, and in both branches, the edge gets removed and
we can count 3 additional clauses from the bookkeeping process. In both branches, we remove 4 degree 3 clauses.
This gives us a branching factor of $\tau(4,4)$. Since we are always doing alternate branching for 
Case 2 (branching at point $c$ and then at point $t$), 
we can apply branching vector addition on $(6,4)$ to $(4,4)$ on both branches to get 
a branching vector of $(10,10,8,8)$. 

Hence, Case 1 takes $O(\tau(6,4)^{m_3})$ time, while Case 2 takes  
 $O(\tau(8,8,10,10)^{m_3})$ time. Since Case 2 is the bottleneck, this gives us
$O(\tau(8,8,10,10)^{m_3})$ $\subseteq$ \\
$O(\tau(8,8,10,10)^{\frac{2}{3}n}) \subseteq O(1.1092^n)$, which absorbs all subexponential terms.
\end{proofsketch}

\subsection{Line 18 of the algorithm}\label{sec-line18}

\noindent
In Line 18, the formula $\varphi$ is left with only degree 2 clauses in the formula. Now suppose
that no simplification rules apply, then we know that the formula must consist 
of cycles
of degree 2 because of Lines 2, 3, 5, 6 and 10 of the algorithm. 
Now if $\varphi$ consists of many components, 
with each being a cycle, then we can handle this by Line 7 of the algorithm. Therefore,
$\varphi$ consists of a cycle.

Now, we choose any variable $x$ in this cycle and branch $x=1$ and $x=0$. Since
all the clauses are of degree 2, we can repeatedly apply Line 10 and other simplification rules
to solve the remaining variables (same idea as in Compression). 
Therefore, we would only need to branch one variable in this line.
This, and repeatedly applying the simplification rules, will only take polynomial time.  

Putting everything together, we have the following result.

\begin{thm} \label{th:twelve}
The whole algorithm runs in $O(1.1120^n)$ time.
\end{thm}

\ifnum\vers=0

\section{Variable-Weighted Counting}

\noindent
One can count not only the overall solutions, but also the solutions
with respect to weights on the literals which are all small nonzero
integers -- for negative weights, one shifts them into positive and
then subtracts at the end, for each possible weight found, a constant.
The weights have to be bounded by a small polynomial $q(n)$ in the number
$n$ of variables. Then every literal
$x$ and $\bar x$ has initially a weight $\vec d(x)$ and
$\vec d(\bar x)$. Now instead of adding and multiplying weights,
one adds and multiplies polynomials in a formal variable $u$ such that
$\vec c(x) = u^k$ says that the term represents one solution in which
$x=1$ and $\vec d(x) = k$. Now for the full assignment $h \in S_\varphi$,
one defines the weight-polynomial to be
$$
   \prod_{\ell: \ \ell \text{ is assigned true in } h} u^{\vec d(\ell)}
$$
and the overall return of the algorithm is the polynomial
$$
   \sum_{h \in S_\varphi} \ \ 
   \prod_{\ell: \ \ell \text{ is assigned true in } h} u^{\vec d(\ell)}
$$
where $S_\varphi$ is the set of solving assingments of the formula $\varphi$.
All updates of the vector $\vec c$ involve only additions and multiplications
and one replaces them by adding and multiplying polynomials in the formal
variable $u$. The result will be a formal polynomial
$$
   \sum_{k=0,1,\ldots,n \times q(n)} a_k \cdot u^k
$$
where $a_0,a_1,\ldots,a_{n \times q(n)}$ are natural numbers whose sum
is at most $2^n$ and the value $a_k$ says that there are exactly $a_k$
solutions in $S_\varphi$ where the sum of all weights of literals which
are $1$ is $k$; as the weights for the literals are multiplied, it means
that the exponents of the formal powers of $u$ of these solutions add up
to $k$. The arithmetics and updating of the polynomials is similar to what
is done for counting pairs of solutions with Hamming distance $k$
for each possible $k$ \cite{HJS19}. All single instructions follow
one of the following steps or a sequence of these steps:
\begin{enumerate}
\item Setting a variable $x$ to a value $b$ after it had been derived
      that $x$ cannot take the value $\bar b$:
      Then one removes $x$ from the list of variables and multiplies
      the overall number of solutions with the polynomial $p(\vec c(\ell))$
      where $\ell = x$ in the case that $b=1$ and $\ell = \bar x$ in the
      case that $b=0$. This is done, for example, in Line 2, where
      several literals are set to the value $0$. Then the multiplication
      there is done explicitly by multiplying the return-polynomial
      from the recursive call with the polynomials generated by fixing
      the literals to $0$.
\item Linking two variables $x$ and $y$, say by setting $y = \bar x$.
      If this is done, one knows that the case $y \neq \bar x$ does
      not occur. Therefore one updates $\vec c(x) = \vec c(x) \times
      \vec c(\bar y)$ and $\vec c(\bar x) = \vec c(\bar x) \times \vec c(y)$.
      The case where $y = x$ is similar. Here the multiplication is not
      done upon returning of a recursive call as in Line 2, but explicitly
      by updating the weight-vector as in Line 5 of the algorithm and in
      the function Link.
\item If one branches a variable $x$ in a formula $\varphi$,
      then the polynomial to be returned
      is just the sum of the one for $\varphi[x=0]$ and the one for
      $\varphi[x=1]$; this is inline with the observation that every
      solution is the solution of exactly one of the formulas
      $\varphi[x=0]$ and $\varphi[x=1]$.
\item If one has only one joint variable $x$ in two components $\psi,\chi$
      of a formula $\varphi=\psi \wedge \chi$,
      then one can contract the easier, say $\chi$, into
      $\psi$, by solving $\chi$ completely under the assumptions
      $x=0$ and $x=1$ and obtaining the result polynomials $q_0$ and $q_1$,
      respectively, and update $\vec c(x) = \vec c(x) \times q_1(u)$
      and $\vec c(\bar x) = \vec c(\bar x) \times q_0(u)$, respectively,
      where $q_0,q_1$ had not yet incorporated the polynomials at
      $\vec c(\bar x)$ and $\vec c(x)$, respectively; the entries of
      the variables only occurring in $\chi$ will be deleted from $\vec c$.
      See Line 14 and the explanations of it for more details. Similarly,
      if there is no joint variable, then the polynomial for the formula
      is just the product of those for $\psi$ and $\chi$, as outlined
      in Lines 6 and 7 in the algorithm. If there
      are more than one connecting variable and one wants to split the
      two components and solve them sepearately, then one first
      branches all variables except one and then second
      contracts $\chi$ into $\psi$. Also these things can be done
      by just adding and multiplying the polynomials.
\item The formula in Line 10 for combining two singletons is also valid
      in the setting of polynomials, the case that one of the two literals
      $x$ or $y$ is $1$ is updated into the case where the resulting
      literal is $1$ and has the weight $\vec c(x) \times \vec c(\bar y) +
      \vec c(\bar x) \times \vec c(y)$, as exactly one of these literals is
      $1$ while the weight of the resulting literal to be $0$ has the weight
      $\vec c(\bar x) \times \vec c(\bar y)$, as this is the case that both
      literals $x,y$ are $0$.
\end{enumerate}
These operations all preserve the invariants; the computation with
polynomials instead of numbers has only an overhead of a polynomial factor.
As the basis of the exponentiation was uprounded in Theorem~\ref{th:twelve},
the corresponding time bound is for this case the same.

\begin{thm}
If the weights of the literals in a variable-weighted X3SAT-formula
are from $\{0,1,\ldots,q(n)\}$ for each $n$-variable instance where
$q$ is a fixed polynomial, then one can count in time $O(1.1120^n)$
how many solutions to the instance have the weight $k$ for each of
$k=0,1,\ldots,n \times q(n)$.
\end{thm}

\noindent
There have been also investigations where the weights are not natural
numbers, but $q(n)$-digit real numbers (better said, rational numbers)
where $q(n)$ is some polynomial (or the number of digits is an extra
parameter). By scaling the measures up, one can assume that they are
natural numbers. Note that this situation is different from the previous
one in the sense that each solution might have a different weight and
therefore there may be exponentially many different solutions and weights.
This would then not allow to count everything in polynomial space.
Therefore the
algorithms for this case are only interested in the number of solution
with maximum (or minimum) weight and not in the overall picture how
the solutions distribute on the different weights. The state of the
art is an algorithm of Porschen and Plagge which runs in $O(1.1193^n)$
time \cite{PP10}. The algorithm of this paper can be adjusted to handle
this problem. In
the main algorithm, there are now two numbers per literal: $\vec c(x)$ is the
number of ``partial solutions'' represented by the 
literal $x$ (which can involve
several original variables due to linking and contracting) and $\vec d(x)$
which is the maximum weight obtained. The updates are now analogous,
except that if there are partial solutions contracted into one
literal, algorithm chooses those which have the maximal weight and adds
up their numbers.
More precisely, the handling is as follows, where
weight $0$ is only taken in the case that there is no correct solution:
\begin{enumerate}
\item If the number of variables in $\varphi$ is small, one can compute the
      return values $(c,d)$ explicitly. For a solution $h$, let
      $$
         \vec c(h) = \prod_{\ell: \ \ell \text{ occurs in } h} \vec c(\ell)
      $$
      and
      $$
         \vec d(h) = \sum_{\ell: \ \ell \text{ occurs in } h} \vec d(\ell).
      $$
      For given $S_\varphi$, let $D = \{\vec d(h): h \in S_\varphi \wedge
      \vec c(h) > 0\}$. If $D$ is not empty then let $d = \max(D)$ and
      $$
         c = \sum_{h \in S_\varphi: \vec d(h)=d} \vec c(h)
      $$
      else let $d=0$ and $c=0$.
      The so obtained pair $(c,d)$ are the return-values for this formula
      $\varphi$.
\item If a formula treated turns out to be unsolvable, then the
      return-values $(c,d)$ are $(0,0)$.
\item If a literal $x$ takes the value $1$ then
      one calls the subroutine with the parameters
      $CountX3SAT(\varphi[x=1],\vec c',\vec d',L,R)$ where $\vec c'$ and
      $\vec d'$ are obtained by omitting the values for $x$ and $\bar x$
      in $\vec c$ and $\vec d$
      and upon receiving the return values $(c,d)$, if $\vec c(x) \times c > 0$
      then one returns $(\vec c(x) \times c,\vec d(x)+d)$ to the main program
      else one returns $(0,0)$ to the main program.
\item If one links $x,y$ by, say, $y = \bar x$, then one drops
      the possibility that $y = x$ and therefore the updates into
      the new values for $x$ are $\vec c(x) = \vec c(x) \times \vec c(\bar y)$,
      $\vec d(x) = \vec d(x)+\vec d(\bar y)$, $\vec c(\bar x) = \vec c(\bar x)
      \times \vec c(y)$, $\vec d(\bar x) = \vec d(\bar x)+\vec d(y)$.
      After that, whenever $\vec c(\ell) = 0$ for a literal $\ell$,
      one makes $\vec d(\ell)=0$ as well.
\item If one branches $x$ then one does the recursive calls to receive
      $(c_0,d_0)$ for $CountX3SAT(\varphi[x=0],\vec c',\vec d',L,R)$ and
      $(c_1,d_1)$ for $CountX3SAT(\varphi[x=1],\vec c',\vec d',L,R)$,
      where $\vec c'$ and $\vec d'$ are obtained by omitting the entries
      for $x,\bar x$ from $\vec c$ and $\vec d$. Now one chooses the
      return values $(c,d)$ according to the first case which applies:
      \begin{enumerate}
      \item If $c_0 \times c(\bar x)+c_1 \times c(x) = 0$
            then one returns $(0,0)$.
      \item If $c_0 \times c(\bar x)=0$ then one returns
            $(c_1 \times c(x),d_1+\vec d(x))$.
      \item If $c_1 \times c(x)=0$ then one returns
            $(c_0 \times c(\bar x),d_0+\vec d(\bar x))$.
      \item If $d_0+\vec d(\bar x)=d_1+\vec d(x)$ then one returns
            $(c_0 \times c(\bar x)+c_1 \times c(x),d_0+\vec d(\bar x))$.
      \item If $d_0+\vec d(\bar x)>d_1+\vec d(x)$ then one returns
            $(c_0 \times c(\bar x),d_0+\vec d(\bar x))$.
      \item If $d_0+\vec d(\bar x)<d_1+\vec d(x)$ then one returns
            $(c_1 \times c(x),d_1+\vec d(x))$.
      \end{enumerate}
\item Assume that $\varphi = \psi \wedge \chi$ whre the formula $\chi$
      is small and easy to evaluate. Furthermore, there is at most
      one common variable $x$ in both formulas.
      In the case that $x$ does not exist, one directly computes
      the return value $(c_\chi,d_\chi)$ of
      $CountX3SAT(\chi,\vec c_\chi,\linebreak[3] \vec d_\chi,L_\chi,R_\chi)$
      with the inputs restricted to $\chi$ for the formula $\chi$ and
      similarly $(c_\psi,d_\psi)$ for the formula $\psi$.
      If $c_\psi \times c_\chi > 0$ then the overall return-values are
      $(c_\psi \times c_\chi,d_\psi+d_\chi)$ else the overall return-values
      are $(0,0)$.
      If $x$ exists and $\chi$ is small, then one computes
      first for $b=0,1$, one let $S_{\chi,b}$ be the solutions
      of $\chi$ with $x=b$ and one computes as in Item 1
      the values $(c_b,d_b)$ for the corresponding case $x=b$.
      Note that $d_b=0$ whenever $c_b = 0$. Then one let
      $\vec c_\psi$ be the restriction of $\vec c$ to $\psi$ and
      $\vec d_\psi$ be the restriction of $\vec d$ to $\psi$ with
      the additional update that $\vec c_\psi(x) = c_1$, $\vec c_\psi(\bar x)
      = c_0$, $\vec d_\psi(x) = d_1$, $\vec d_\psi(\bar x) = d_0$.
      Now the return-values of this case are the
      output of $CountX3SAT(\psi,\vec c_\psi,\vec d_\psi,L_\psi,R_\psi)$.
\item When contracting two singleton literals $x,y$ into one literal
      in Line 10, then the new literal -- here called $z$ -- will take
      the following values in $\vec c$ and $\vec d$, always according
      to the first case which applies:
      \begin{enumerate}
      \item If $\vec c(x) \times \vec c(\bar y)+
                \vec c(\bar x) \times \vec c(y)=0$
            then $\vec c(z)=0$ and $\vec d(z)=0$;
      \item If $\vec c(x) \times \vec c(\bar y)=0$
            then $\vec c(z)=\vec c(\bar x) \times \vec c(y)$
            and $\vec d(z)=\vec d(\bar x)+\vec d(y)$;
      \item If $\vec c(\bar x) \times \vec c(y)=0$
            then $\vec c(z)=\vec c(x) \times \vec c(\bar y)$
            and $\vec d(z) = \vec d(x)+\vec d(\bar y)$;
      \item If $\vec d(x)+\vec d(\bar y) = \vec d(\bar x)+\vec d(y)$
            then $\vec c(z) = \vec c(x) \times \vec c(\bar y)+
                \vec c(\bar x) \times \vec c(y)$
            and $\vec d(z) = \vec d(x)+\vec d(\bar y)$;
      \item If $\vec d(x)+\vec d(\bar y)>\vec d(\bar x)+\vec d(y)$
            then $\vec c(z)=\vec c(x) \times \vec c(\bar y)$
            and $\vec d(z) = \vec d(x)+\vec d(\bar y)$;
      \item If $\vec d(x)+\vec d(\bar y)<\vec d(\bar x)+\vec d(y)$
            then $\vec c(z)=\vec c(\bar x) \times \vec c(y)$
            and $\vec d(z)=\vec d(\bar x)+\vec d(y)$;
      \item If $\vec c(\bar x) \times c(\bar y) > 0$
            then $\vec c(\bar z) = c(\bar x) \times c(\bar y)$
            and $\vec d(\bar z) = d(\bar x)+d(\bar y)$
            else $\vec c(\bar z) = 0$ and $\vec d(\bar z) = 0$.
      \end{enumerate}
      After adding the entries of $z,\bar z$ into $\vec c,\vec d$
      as above, one removes the entries of $x,\bar x,y, \bar y$
      from $\vec c,\vec d$ and replaces $x \vee y$ by $z$ in
      $\varphi$ and calls, with these updated parameters,
      $CountX3SAT(\varphi,\vec c,\vec d,L,R)$ and passes
      the return-values on to the main program.
\end{enumerate}
Dahll\"of gives in his dissertation \cite{Dah06} an outline of this method.
Again, the only modification of the main algorithm is the handling of the
data structure to do the book keeping for the maximum weight of the subproblem
summarised in the current literal $x$ which is $\vec d(x)$ and the number
of subtuples belonging to this weight stored in $\vec c(x)$. As this
overhead is only a polynomial factor, again the runtime is the same.

\begin{thm}
One can count the number of maximal solutions of a variable-weighted
X3SAT instance of $n$ variables in time $O(1.1120^n)$.
\end{thm}

\section{Conclusions}

\noindent
In this paper, we gave an algorithm to solve the \#X3SAT problem in
$O(1.1120^n)$. The novelty in this paper is to use the Monien and Preis
result to help us to deal with degree 3 clauses. We used also for
the Monien and Preis part the technique of branching factors to analyse
the search tree while branching the variables in the bisection instead
of the usual method of counting the number of variables involved to brute
force. Doing so allows us to tighten our analysis much more.

We also observe that the same algorithm, with only minor adjustments to the
bookkeeping of the number of solutions, allows for integer-weighted X3SAT
instances where the weights are bounded by a fixed polynomial $q(n)$ with
$n$ being the number of variables, to count the number of solutions for
each possible weight with a time-usage which is only by a polynomial
factor larger than the original algorithm. Furthermore, if the weights
can have exponential size, then we follow Dahll\"of's approach of
counting only the maximum weight solutions in order to keep the algorithm
in polynomial space \cite{Dah06}.

Counting problems are usually much harder than their decision problem
counterpart. Wahlstr\"om gave an algorithm to decide X3SAT in
$O(1.0984^n)$ \cite{Wah07} and is currently the fastest exact algorithm
for this problem. With our algorithm, the difference
in time between the decision problem and the counting problem have narrowed 
significantly. However, narrowing the gap more might prove to be difficult,
as the most optimised X3SAT algorithms use rules which are not compatible
with counting like, for example, \cite[Transformation (23)]{BMS05}.
For that reason, we came up with our own DPLL style branching frontend
and the Monien Preis part at the end still allows some improvement
in the frontend which is the current bottleneck of the algorithm.
\fi

\end{document}